\documentclass[review]{elsarticle}

\usepackage{lineno,hyperref}
\usepackage{setspace}
\usepackage{amsfonts}
\usepackage{amssymb}
\usepackage{amsmath}
\usepackage{stmaryrd}
\usepackage{bbding}
\usepackage{algorithm}
\usepackage{algorithmic}
\usepackage{setspace}
\usepackage{amsfonts}
\usepackage{amssymb}
\usepackage{stmaryrd}

\usepackage{bbding}
\usepackage{algorithm}
\usepackage{algorithmic}
\usepackage{subfigure}

\usepackage{mathrsfs}
\usepackage{graphicx}
\usepackage[justification=centering]{caption}
\usepackage{subfigure}
\usepackage{comment}

\newtheorem{mydef}{Definition}
\newtheorem{mytheorem}{Theorem}

\newtheorem{proof}{Proof}

\journal{Journal of \LaTeX\ Templates}

\begin{document}

\begin{frontmatter}

\title{Differentially Private User-based Collaborative Filtering Recommendation Based on K-means Clustering}


\author[mymainaddress]{Zhili Chen}
\cortext[mycorrespondingauthor]{Corresponding author}
\ead{zlchen@ahu.edu.cn}

\author[mymainaddress]{Yu Wang}

\author[mymainaddress]{Shun Zhang}

\author[mymainaddress]{Hong Zhong}

\author[mymainaddress2]{Lin Chen}

\address[mymainaddress]{School of Computer Science and Technology, Anhui University, 230601 Hefei, China}
\address[mymainaddress2]{Lab. Recherche Informatique (LRI-CNRS UMR 8623), Univ. Paris-Sud, 91405 Orsay, France}

\begin{abstract}
Collaborative filtering (CF) recommendation algorithms are well-known for their outstanding recommendation performances, but previous researches showed that they could cause privacy leakage for users due to k-nearest neighboring (KNN) attacks. Recently, the notion of differential privacy (DP) has been applied to privacy preservation for collaborative filtering recommendation algorithms. However, as far as we know, existing differentially private CF recommendation schemes degraded the recommendation performance (such as recall and precision) to an unacceptable level. In this paper, in order to address the performance degradation problem, we propose a differentially private user-based collaborative filtering recommendation scheme based on k-means clustering (KDPCF). Specifically, to improve the recommendation performance, we first cluster the dataset into categories by k-means clustering and appropriately adjust the size of the target category to which the target user belongs, so that only users in the well-sized target category can be used for recommendations. Then we efficiently select a set of neighbors from the target category at one time by employing only one exponential mechanism instead of the composition of multiple ones, and base on the neighbor set to recommend. We theoretically prove that our scheme achieves differential privacy. Empirically, we use the MovieLens dataset to evaluate our recommendation system. The experimental results demonstrate a significant performance gain compared to existing schemes.
\end{abstract}

\begin{keyword}
Differential Privacy, K-means Clustering, Recommendation System, Collaborative Filtering
\end{keyword}

\end{frontmatter}


\section{Introduction}

\subsection{Motivation and Problem Statement}

Today's Internet is inundated by data, making it more and more difficult for users to quickly locate the data they look for. Recommendation systems are thus developed to facilitate the searching task by recommending pertinent data or resources for users. A class of efficient and personalized recommendation systems are based on collaborative filtering (CF) algorithms, which discover the potential consumption trend of users by mining the user-data relationships. Despite the performance gain of CF-based recommendation algorithms, the involvement of users' historical data in the recommendation process may cause severe information leakage, as demonstrated by the K-nearest neighboring (KNN) attacks \cite{mobasher2005effective}.


\subsection{Limitation of Prior Art}

For individual privacy protection issues, Dwork et al. first proposed the differential privacy (DP) protection method in \cite{dwork2006calibrating,dwork2008differential}, which can well deal with the differential attack and other attacks based on background knowledge. In view of its strong privacy guarantees, many articles have recently applied the concept of DP to recommendation systems to solve the privacy issues of recommendation systems. But the simple application of DP often leads to the loss of data utility, namely, it greatly reduces the recommendation performance. How to achieve the balance between privacy protection and recommendation accuracy is an important problem.

Previous differentially private user-based collaborative filtering recommendation schemes suffer a great performance degradation when applying the notion of differential privacy \cite{zhu2013differential} \cite{zhu2016differential}. The authors of \cite{zhu2013differential} found that the Mean Absolute Error (MAE) values of recommendations degrade much after applying differential privacy techniques, and used local sensitivities instead of global sensitivities in their scheme to improve the recommendation performance. However, applying local sensitivities weakens differential privacy guarantees \cite{dwork2014algorithmic}. The authors of \cite{zhu2016differential} applied the advanced composition theorem to improve the recommendation performance, but their experimental results demonstrated that the recommendation recalls with differential privacy degraded severely compared to the original ones. Actually, these performance degradations for recommendations may be unacceptable in practice.

For the above performance degradation issue, there are probably two main causes. The first one is that, when applying the exponential mechanism, too many unrelated or weakly related output items are defined for the differentially private recommendations, and these output items may overwhelm the accurate recommendation results due to the randomness of the exponential mechanism. The second one is that too much ``noise'' is added to the recommendations due to the composition of a great many differentially private algorithms. Having confirmed these causes, we refer to the practice of Xue et al. \cite{xue2005scalable}, and apply the clustering algorithm to designing a differentially private recommendation system. Also, we reduce the application times of the exponential mechanism as much as possible, to improve the recommendation performance.



\subsection{Summary of Contribution}

In this paper, the k-means clustering algorithm is applied to a differentially private user-based collaborative filtering recommendation. First of all, we cluster the users from the entire rating matrix $M$, and adjust the size of the target category, which contains the target user. Next, we select a neighbor set from the target category at one time, using the exponential mechanism. Finally, predict user ratings and give top-$m$ items recommendations based on the neighbor set selected.

The contributions of this paper can be summarized as follows:

\begin{enumerate}
\item We propose a differentially private user-based collaborative filtering recommendation scheme based on k-means clustering. Theoretically, we prove that our scheme satisfies $\epsilon$-differential privacy. Also, we have experimentally shown that our scheme provides a higher recommendation performance than an existing scheme.


\item We come up with an approach based on the bisecting k-means to adjust the size of the target category, such that the target category is in an appropriate size, and both hence recommendation performance and user privacy are reasonably balanced.


\item We design a random neighbors selection algorithm to select a neighbor set at one time with only an application of the exponential mechanism instead of multiple ones. This algorithm adds much less noise to the proposed mechanism compared to using multiple exponential mechanisms, and thus improves the recommendation performance.
\end{enumerate}

The remainder of this paper is organized as follows.
In Section~\ref{sec:relatedwork}, we briefly review the related work. Section~\ref{sec:Problem Statement} describes the privacy issues of the user-based collaborative filtering recommendations. Section~\ref{sec:Technical Preliminaries} introduces some basic concepts and theorems used in this work. We provide the detailed design of our scheme in Section~\ref{sec:KDPCF}. Section~\ref{sec:experiment} discusses and analyzes our experimental results. Finally, we conclude the paper and provides future work in Section~\ref{sec:conclusion}.

\section{Related Work}\label{sec:relatedwork}

\textbf{Traditional Privacy-preserving Recommendations.} In recommendation systems, the traditional approaches to protecting user privacy are mainly based on cryptography and anonymization. The cryptography based approach is suitable for recommendations performed by multiple mutually distrustful parties, and it normally incurs heavy computational costs \cite{canny2002collaborative,armknecht2011efficient}, especially when the amount of historical data is large. Instead, the anonymization based approach has much higher computational efficiency, but it cannot resist background knowledge attacks, and may weaken the utility of the data severely (especially for high-dimensional data) \cite{brickell2008cost,chang2010towards}. Different from these traditional works, in this paper we focus on achieving differential privacy for recommendation systems, such that the computational cost is low and the user privacy is theoretically guaranteed.

\textbf{Differentially Private Recommendations.} The notion of differential privacy has been applied to recommendation systems previously. McSherry et al. \cite{mcsherry2009differentially} first applied differential privacy to a recommendation system, and proved that it is feasible to achieve differential privacy for the recommendation system without serious loss of recommendation accuracy. Hua et al. \cite{hua2015differentially} analyzed the privacy threats in various situations, concerning whether the recommender is trusted and whether the recommender is online, and gave corresponding differentially private recommendation schemes based on matrix factorization. Meng et al. \cite{meng2018personalized} divided users' privacy into sensitive privacy and non-sensitive privacy, allowing users to classify their privacy and add noise to their data locally in recommendations. Similarly, literatures \cite{zhu2013differential} and \cite{feng2016differential} improved the recommendation performance by considering recommendation-aware sensitivity or different levels of privacy. In this work, we aim to improve the recommendation performance of differentially private user-based collaborative filtering recommendations mainly based on k-means clustering.

The closest work to ours is the one carried out by Zhu et al. \cite{zhu2016differential}. They apply the exponential mechanism to repeatedly select the recommended items in both item-based and user-based collaborative filtering recommendations. Moreover, in order to improve the recommendation performance, they compared multiple similarity functions and selected the best one as the utility function for the exponential mechanism. However, the overall recommendation performance of the work \cite{zhu2016differential} still degrades very much. In our work, we focus on designing a differentially private user-based collaborative filtering recommendation scheme, and further improve the recommendation performance by employing a preprocessing based on k-means clustering and by reducing the number of applications of exponential mechanism to merely one.


\section{Problem Statement}\label{sec:Problem Statement}

\subsection{User-based Collaborative Filtering Recommendation}

We consider a user-based collaborative filtering recommendation \cite{herlocker2004evaluating,sarwar2001item}, where users submit their historical rating scores on items to a recommender who then predicts a target user's future rating scores by learning the rating similarities between the user and others from the historical data. Let $M$ denote the $|U|\times|I|$ user$-$item rating matrix, where $U$ is the user vector and $I$ is the item vector in the recommendation, $M_{ui}$ represents the historical rating score on item $i\in I$ given by user $u\in U$ if it is a non-zero value, and represents an unrated score otherwise. In practice, matrix $M$ is usually sparse, and the goal of the collaborative filtering recommendation is to predict these missing $M_{ui}$.

The basic process of user-based collaborative filtering recommendation is as follows, shown in Algorithm 1. First, the rating similarities between the target user $u$ and other users are calculated (line 1), then neighboring users are selected based on the similarities, and user $u$'s unrated scores are predicted with its neighboring users (2-5 lines), and finally, the highest rated $m$ items will be recommended to user $u$ (line 6). The main ideas of the CF algorithm can be summarized as: Users with similar preferences recommend items each other. Thus, how to calculate rating similarities is an important issue directly related to the accuracy of recommendation results.

\begin{algorithm}[t]
\caption{User-based collaborative filtering recommendation} 
{\bf Input:} 
\hangafter 1
\hangindent 1.5em
\noindent
User-item rating matrix $M$, recommendation list length $m$, neighbor set size $N$, target user $u$, and its rating record $I_u$.

{\bf Output:} 
\hangafter 1
\hangindent 1.5em
\noindent
User $u$'s recommendation list $R_u$.

    \begin{algorithmic}[1]
           \STATE Calculate similarities between user $u$ and all other users according to $M$.

           \FORALL{$i\in I - I_u$}

           \STATE Find user $u$'s neighbor set $\mathcal{N}_u$ with top $N$ users in term of similarities.

           \STATE Predict all of user $u$'s unrated scores based on its neighbor set $\mathcal{N}_u$.

           \ENDFOR

           \STATE Add the top $m$ predicted items to user $u$'s recommendation list $R_u$.

    \end{algorithmic}
\end{algorithm}

\subsection{Privacy Problem}

In our user-based collaborative filtering (CF) recommendations, the recommender is assumed to be trusted. A malicious user may receive recommendation results from the recommender, and analyze these results to infer the information of other users. Thus, the privacy issue mainly comes from the selection of a target user's neighbors. In order to make the recommendation results more accurate, a user-based CF recommendation tends to select the neighbor who is the most similar to the target user. However, it is this way of selecting the most similar neighbor that makes a privacy breach in user-based CF recommendations.

Table~\ref{tab:privacy-issue} illustrates an example on the privacy issue mentioned above. Initially, the adversary knows the recommendation algorithm and a user $u_3$'s rating record $I_3 = \{0, 3, 5, 4, 0, 0\}$, where each ``0'' represents an unrated item. Some time later, user $u_3$ rated an item $i_5$ and we have $I_3 = \{0, 3, 5, 4, ?, 0\}$, where the ``?'' represents a new rating score unknown to the adversary. Now, the adversary can launch an attack to find which item has been rated newly as follows. The adversary first registers a sibyl user while keeping the sibyl user's rating record consistent with user $u_3$'s historical record, i.e. $I_s = \{0, 3, 5, 4, 0, 0\}$. Next, the sibyl user sends a recommendation request to the system, and according to the principle of the user-based CF algorithm, the sibyl user will have a great chance of having user $u_3$ selected as its neighbor. Then, based on the recommendation results of the system, the adversary can easily deduce which item newly rated by $u_3$, and if the adversary can further learn the change in the similarity between the sibyl and $u_3$, the adversary can even deduce what is the specific rating score of the item $i_5$.
\begin{table}
  \centering
  \caption{Privacy Issue of User-based Collaborative Filtering Recommendations}\label{tab:privacy-issue}
\begin{tabular}{|c|c|c|c|c|c|c|}
\hline
 & $i_1$ & $i_2$ & $i_3$ & $i_4$ & $i_5$ & $i_6$\\
\hline
$u_1$ & 2 &  & 4 &  & 5 & \\
\hline
$u_2$ & 3 & 4 &  & 5 &  & 4\\
\hline
$u_3$ &  & 3 & 5 & 4 & ? & \\
\hline
$u_4$ & 2 & 3 &  & 3 &  & 5\\
\hline
sibyl &  & 3 & 5 & 4 &  & \\
\hline
\end{tabular}
\end{table}

\subsection{Design Target}

The above privacy issue has been basically addressed before \cite{zhu2016differential}. However, the previous work suffer a great degradation of recommendation performance while guaranteeing user privacy, and thus is far from practical applications. The design target of this paper is to provide a better balance between the recommendation performance and the user privacy. Specifically, we aim to provide a significant improvement of the recommendation performance under the same privacy levels.

\section{Technical Preliminaries}\label{sec:Technical Preliminaries}

\subsection{Differential Privacy}
Differential privacy (DP) is a new privacy model based on data distortion \cite{dwork2014algorithmic}. By adding controllable noise to the original datasets, DP guarantees the protection of personal information in a provable way, while ensuring that the perturbed data have similar statistical properties to the original data. Specifically, DP guarantees that adding or deleting any record to or from a dataset will not have significant influence on the query result based on the dataset. Hence, even if an adversary knows all records except the sensitive one in the dataset, the protection of the sensitive record can be still guaranteed.

\begin{mydef}
($\epsilon$-Differential Privacy)\cite{dwork2008differential,dwork2014algorithmic}: A randomized algorithm $\mathcal{M}$ satisfies $\epsilon$-differential privacy if for any two neighbouring datasets $t$ and $t'$ differing on at most one element, and for any set of outcomes $R\subseteq Range(\mathcal{M})$, $\mathcal{M}$ satisfies:

\begin{equation}
P[\mathcal{M}(t)\in R]\leq exp(\epsilon)\times P[\mathcal{M}(t')\in R].
\end{equation}

Where $\epsilon$ is the privacy budget, which decides the privacy level of the algorithm. A greater $\epsilon$ means less noise added and thus a lower privacy level.
\end{mydef}

\subsection{Exponential Mechanism}

Exponential mechanism (EM) is a common technique for designing algorithms with differential privacy. It was developed by Frank McSherry and Kunal Talwar \cite{mcsherry2007mechanism}. EM is popularly used, for it can well deal with both numeric and non-numeric problems.

\begin{mydef}
(Exponential Mechanism)\cite{mcsherry2007mechanism}: Given a quality function $q(t,r):\mathbb{N}^{|\chi|}\times \mathcal{R}\rightarrow \mathbb{R}$, $\triangle q$ is the sensitivity of quality function. The exponential mechanism $\mathcal{M}$ selects and outputs an element $r \in R$ with probability proportional to $exp(\frac{\epsilon q(t,r)}{2\triangle q})$.
\end{mydef}

\begin{mydef}
(Sensitivity of Exponential Mechanism)\cite{dwork2014algorithmic}: The sensitivity of exponential mechanism is defined as follows:

\begin{equation}
\triangle q=\max_{r \in R}\max_{t,t':\|t-t'\|_1\leq1}\mid q(t,r)-q(t',r)\mid
\end{equation}

Where $q(t,r)$ is a quality function, $r\in R$ is a valid output of $q(t,r)$. The sensitivity of the function q indicates the greatest effect of a single data change on the output. The size of $\triangle q$ will also directly affect the amount of noise introduced by the algorithm.
\end{mydef}

\subsection{k-means Clustering}
The k-means clustering algorithms have been used in recommendation systems without privacy guarantees in order to improve the recommendation efficiency \cite{macqueen1967some,arthur2009k}. After clustering, the neighbor selection of a target user is limited in the same category as the user, and thus the computational cost can be greatly reduced, especially when the user-item rating matrix is large. In our case, we apply k-means clustering algorithms to improving the sampling quality in the exponential mechanism, and thus improve the recommendation performance.

To get a good clustering result using the k-means clustering algorithms, two kinds of parameters should be determined properly in advance: the number of clusters $k$ and the initial cluster centers. The $k$ value is related to specific datasets and usually determined by approaches based on Silhouette Coefficient \cite{rousseeuw1987silhouettes:} or Elbow method \cite{ketchen1996the}. In our work, we will adjust the $k$ value to meet our specific requirements. It has been shown that the initial clustering centers should be selected uniformly to get a good clustering result \cite{arthur2007k}. Thus, we will use the k-means++ algorithm given by D.Arthur to determine the initial cluster center, as shown in Algorithm 2.

\begin{algorithm}[t]
\caption{The k-means++ algorithm \cite{arthur2007k}}
{\bf Input:}
\hangafter 1
\hangindent 1.5em
\noindent
The amount of cluster k.

{\bf Output:}
\hangafter 1
\hangindent 1.5em
\noindent
k initial cluster center $c_1,c_2,...,c_k$.

    \begin{algorithmic}[1]
           \STATE Choose a user randomly as the first initial cluster center $c_1$.

           \FORALL{$i=2:k$}

           \STATE Calculate the shortest distance $D(u)$ between each user and all current cluster centers.

           \STATE Sample every user $u\in U$ with probability $P(u)$ as the next cluster center $c_i$.

                    \[
                        P(u)=\frac{D(u)^2}{\sum_{u\in U}D(u)^2}
                    \]

           \ENDFOR

    \end{algorithmic}
\end{algorithm}

\section{Our Scheme}\label{sec:KDPCF}

\subsection{Overall Design}
\begin{figure}[htbp]
    \centering
    \includegraphics[width=2.2in]{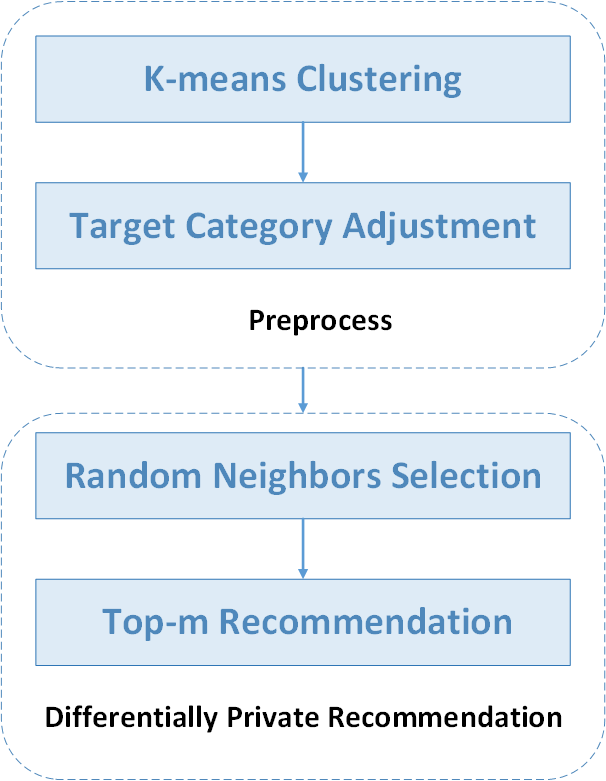}
    \caption{KDPCF Algorithm Flow Diagram}\label{fig:1}
\end{figure}

To obtain a better balance between recommendation performance and user privacy, we design a differentially private user-based collaborative filtering recommendation scheme based on k-means clustering (KDPCF). The main process of KDPCF is shown in Figure 1, and the design rationale for every step is described as follows.

\textbf{Step 1. K-means Clustering.} In the rating matrix, the quantity of users is normally large, but usually only a small part of users are significant to the recommendation for a target user. Through k-means clustering, such part of users (called the target category) can be found, and the recommendation can be based upon these users.

\textbf{Step 2. Target Category Adjustment.} The target category resulted from k-means clustering may be in various sizes. Inappropriate category size (too big or too small) can cause performance or privacy issues. Thus, this step adjusts the target category with two size thresholds, $C_{max}$ and $C_{min}$, so that the category size falls in $[C_{min}, C_{max}]$.

\textbf{Step 3. Random Neighbors Selection.} In this step, the exponential mechanism is employed to select the neighbor set of a target user with differential privacy. To achieve a good recommendation performance, instead of selecting the neighbors one by one with repeated applications of the exponential mechanism, a neighbor set is selected at one time with an application of the exponential mechanism.

\textbf{Step 4. Top-m Recommendation.} Based on the neighbor set selected in the previous step, this step predicts the rating scores of the unrated items of the target user, and finally recommend the top-$m$ items in term of the predicted scores to the user.


\subsection{The Detailed Design}

The algorithmic design of KDPCF can be presented in Alg.~\ref{alg:kdpcf}. The detailed design is described and discussed as follows.

\begin{algorithm}[h]
\caption{Differentially Private User-based Collaborative Filtering Recommendation based on K-means Clustering}\label{alg:kdpcf}
\begin{algorithmic}[1]

\REQUIRE User-item rating matrix $M$, recommendation list length $m$, privacy parameter $\epsilon$, target user $u$, its rating record $I_u$ and neighbor set size $N$;

\ENSURE $u$'s recommendation list $R_u$;

\textbf{Step 1. K-means Clustering.}

\STATE Perform k-means clustering on matrix $M$, and get the target category $C^*$;

\textbf{Step 2. Target Category Adjustment.}

\STATE Adjust category $C^*$ in size, and get the corresponding rating matrix $M^*$;

\textbf{Step 3. Random Neighbors Selection.}

\STATE Calculate the similarities between target user $u$ and other users $v \in C^*$: \label{line:cal-sim}

\FORALL{($v\in C^* \wedge v \neq u$)}

\STATE
\[
Sim(u,v) \leftarrow \frac{\sum_{i\in I_{uv}}(r_{ui}-\overline{r}_u)(r_{vi}-\overline{r}_v)}{\sqrt{\sum_{i\in I_{uv}}(r_{ui}-\overline{r}_u)^2\sum_{i\in I_{uv}}(r_{vi}-\overline{r}_v)^2}}
\]

\ENDFOR \label{line:endfor1}

\STATE Sample a random user set $U^*$ from $C^*$, and calculate the probability distribution on the set $\mathbb{N}$ of all possible neighbor sets of size $N$ from the set $U^*$ as follows:\label{line:cal-distr}

\FOR{$\mathcal{N} \in \mathbb{N}$}

\STATE
\[
\frac{exp(\epsilon\sum_{v\in \mathcal{N}}|Sim(u,v)|/(2\Delta q))}{\sum_{\mathcal{N'} \in \mathbb{N}}exp(\epsilon\sum_{v'\in \mathcal{N'}}|Sim(u,v')|/(2\Delta q))}
\]
\ENDFOR \label{line:endfor2}

\STATE Select a neighbour set $\mathcal{N}_u \in \mathbb{N}$ with the above probability distribution; \label{line:select}

\textbf{Step 4. Top-m Recommendation.}

\STATE Compute user $u$'s recommendation list $R_u$ of length $m$ from $\mathcal{N}_u$;

\end{algorithmic}
\end{algorithm}

\textbf{Step 1. K-means Clustering.} This step applies k-means clustering described in Section 4.3 directly to the overall user-item rating matrix $M$. The matrix $M$ is normally very large, e.g., it may involve thousands of users, and hundreds of items. Among these user ratings, usually only the ratings of a small part of users are significant to the recommendation for a target user. The idea is that k-means clustering can help to find out such small part of users. In short, this step simply uses the similarity defined in Eq.~\eqref{equ:sim} to do k-means clustering to matrix $M$, and finds out the target category which contains the target user.

\textbf{Step 2. Target Category Adjustment.}

The purpose of applying k-means clustering is to limit the neighbor selection of the target user to the target category, and thus improve the recommendation performance. However, due to the randomness of k-means clustering, the resulted target category may be of various sizes. The target category size may have a significant impact on both recommendation performance and user privacy. Specifically, if the target category is much larger than the neighbor set size, then the performance enhancement of the recommendation will be very limited, since there are still a large quantity of users in the category which will overwhelm the neighbor selection. On the contrary, if the target category is close to or even less than the neighbor set size, the user privacy protection will be difficult to achieve, since there are too few neighbor sets to be selected, and the selection tends to be more deterministic.

To address the category size issue above, this step adjusts the target category as follows. First, it appropriately sets a pair of size thresholds $C_{\max}$ and $C_{\min}$. Then, it compares the size of the target category with the two thresholds: if the category size is greater than $C_{\max}$, it uses the target user as one of the clustering centers, and apply Algorithm 2 to perform bisecting k-means clustering again; if the category size is less than $C_{\min}$, it merges the target category with its nearest category. This process repeats until the target category size falls in the range from $C_{\min}$ to $C_{\max}$. Note that the size thresholds $C_{\min}$ and $C_{\max}$ can be set empirically. For example, in our experiments, we set $C_{\max}$ and $C_{\min}$ to $10N$ and $5N$, respectively, where $N$ is the size of neighbor sets.

\textbf{Step 3. Random Neighbors Selection.}

Having found the target category with an appropriate size, this step turns to the Random Neighbors Selection with differential privacy based on the target category. The main idea is as follows. The algorithm first computes the similarities between the target user $u$ and all other users in the target category. Then, based on these similarities, the algorithm randomly selects a neighbor set of the target user from the target category with the exponential mechanism, achieving differential privacy.

In recommendations, Pearson correlation coefficients are widely used to calculate similarities between users. Let $I_u$ and $I_v$ be the rating records of users $u$ and $v$, respectively, and a Pearson correlation coefficient can be used to calculate the similarity between $u$ and $v$ as Eq.~\eqref{equ:sim}
\begin{equation}\label{equ:sim}
Sim(u,v)\!=\!\frac{\sum_{i\in I_{uv}}(r_{ui}-\overline{r}_u)(r_{vi}-\overline{r}_v)}{\sqrt{\sum_{i\in I_{uv}}(r_{ui}-\overline{r}_u)^2\sum_{i\in I_{uv}}(r_{vi}-\overline{r}_v)^2}}
\end{equation}
where $\overline{r}_u$, $\overline{r}_v$ represent the average rating sores of users $u$ and $v$, respectively (i.e. $\overline {r}_u = \sum_{i\in I_u}r_{ui}/|I_u|$, and $\overline {r}_v = \sum_{i\in I_v}r_{vi}/|I_v|$), $I_{uv}$ represents the intersection of rating records of users $u$ and $v$. It is noteworthy that $Sim(u, v)\in [-1,1]$, and the absolute value of $sim(u,v)$ indicates the strength of the correlation between users $u$ and $v$. Moreover, when the similarity is positive (resp. negative), it indicates that the two users are positively (resp. negatively) related.

In this step, we use Pearson correlation coefficients to measure similarities. Specifically, in Alg.~\ref{alg:kdpcf}, Lines from \ref{line:cal-sim} to \ref{line:endfor1} calculate the similarities between the target user $u$ and other users in the target category using Eq.~\eqref{equ:sim}. These similarities provide the basis for Random Neighbors Selection.

To protect user privacy, we want to make the Random Neighbors Selection satisfy differential privacy. Furthermore, we also care for the recommendation performance achieved. Thus, we need to apply the exponential mechanism to the neighbor selection and carefully design the quality function. Our main idea is to select the neighbor set with one application of the exponential mechanism.



Given the target category $C^*$, the target user $u$, and a user set $\mathcal{N} \subseteq C^* - \{u\}$, the quality of selecting the set $\mathcal{N}$ as user $u$'s neighbor set is given by:
\begin{equation}\label{equ:q}
    q(C^{*}, u, \mathcal{N})=\sum_{v \in \mathcal{N}}|Sim(u, v)|
\end{equation}

Here, we consider both positive and negative correlations the same, and define the quality function as the sum of absolute values of similarities. According to the exponential mechanism, the probability of outputting the set $\mathcal{N}$ as the neighbor set is proportional to
\begin{equation}
\exp(\frac{\epsilon q(C^*, u, \mathcal{N})}{2\Delta q})
\end{equation}
where $\Delta q$ is the sensitivity of quality function $q$, which can be computed as Eq.~\eqref{equ:deltaq}.
\begin{equation}\label{equ:deltaq}
\Delta q=\!\max_{\mathcal{N}}\max_{\|C^{*}_1 \!-\! C^{*}_2\|_1\leq1}\!\mid \!q(C^{*}_1, u, \mathcal{N})-q(C^{*}_2, u, \mathcal{N})\!\mid =\!1
\end{equation}
where $C^{*}_1$ and $C^{*}_2$ are two adjacent target categories that only differ in one user.

In the Random Neighbors Selection, we apply the quality function defined in Eq.\eqref{equ:q} to the exponential mechanism, and select the neighbor set of the target user at one time. Specifically, in Alg.~\ref{alg:kdpcf}, Lines from \ref{line:cal-distr} to \ref{line:endfor2} sample a random user set $U$ from the target category $C^*$, enumerate all possible neighbor sets of size $N$ from the user set $U$ to form a set $\mathbb{N}$ of neighbor sets, and calculate the probability distribution over the set $\mathbb{N}$. Line~\ref{line:select} simply selects a neighbor set with the probability distribution calculated. There are two points noteworthy as follows.

(1) Instead of selecting one neighbor with an application of the exponential mechanism repeatedly, selecting a neighbor set at one time can add less noise to the selection and obtain a quality function with less sensitivity, and thus achieve better recommendation performance.

(2) For two adjacent categories, we mean that both categories contain exactly the same $N$ users and only one single user differs in its rating scores. That is, we use the bounded differential privacy definition \cite{li2016differential}.

\textbf{Step 4. Top-m Recommendation.}

Given the neighbor set $\mathcal{N}^*$ selected, this step constructs the recommendation list $R_u$ by collecting the rated items of all neighbors in $\mathcal{N}^*$ while excluding the rated items from $I_u$, and then predicts the rating sores of the items from $R_u$ Eq.~\eqref{equ:sore}.
\begin{equation}\label{equ:sore}
    r^*_{ui}=\overline{r}_u + \frac{\sum_{v \in \mathcal{N}^*} Sim(u, v)*(r_{vi} - \overline{r}_v)}{\sum_{v \in \mathcal{N}^*}|Sim(u, v)|}
\end{equation}
where $\overline{r}_u$, $\overline{r}_v$ represent the average rating sores of users $u$ and $v$, respectively. Finally, in term of rating scores $r^*_{ui}$, the top-$m$ items from $R_u$ are returned as the recommendation list for user $u$.



\subsection{Privacy Analysis}
We state the differential privacy of KDPCF in Theorem~\ref{the:dp}.
\begin{mytheorem}\label{the:dp}
KDPCF (as described in Alg.~\ref{alg:kdpcf}) satisfies $\epsilon$-differential privacy.
\end{mytheorem}

\begin{proof}
Let $T_1$ and $T_2$ be a pair of adjacent datasets, that is, they contain exactly the same users and only one single user differs in its rating scores, and $t_1$ and $t_2$ represent the target categories of them, respectively. We regard the steps of k-means clustering and target category adjustment as a preprocess for each target user $u$, that is, we let the differing happens after the clustering, and hence both target categories still contain exactly the same users, and they differ in at most a single user in rating scores. According to the exponential mechanism, we can derive the probability ratio of any output $\mathcal{N}$ of Random Neighbors Selection step as follows:
\begin{align*}
  \frac{Pr[M_{RNS}(t_1)=\mathcal{N}]}{Pr[M_{RNS}(t_2)=\mathcal{N}]}& = \frac{Pr(t_1, \mathbb{N}) \cdot \frac{exp(\frac{\epsilon q(t_1, \mathcal{N})}{2\Delta q})}{\sum_{\mathcal{N}^{'}\in \mathbb{N}}exp(\frac{\epsilon q(t_1,\mathcal{N}^{'})}{2\Delta q})}}{Pr(t_2, \mathbb{N}) \cdot \frac{exp(\frac{\epsilon q(t_2,\mathcal{N})}{2\Delta q})}{\sum_{\mathcal{N}^{'}\in \mathbb{N}}exp(\frac{\epsilon q(t_2,\mathcal{N}^{'})}{2\Delta q})}}          \\
  &                                          =\left(\frac{exp(\frac{\epsilon q(t_1,\mathcal{N})}{2\Delta q})}{exp(\frac{\epsilon q(t_2,\mathcal{N})}{2\Delta q})}\right)\cdot\left(\frac{\sum_{\mathcal{N}^{'}\in \mathbb{N}}exp(\frac{\epsilon q(t_2,\mathcal{N}^{'})}{2\Delta q})}{\sum_{\mathcal{N}^{'}\in \mathbb{N}}exp(\frac{\epsilon q(t_1,\mathcal{N}^{'})}{2\Delta q})}\right)    \\
  &                                                 \leq exp(\frac{\epsilon}{2})\cdot\left(\frac{\sum_{\mathcal{N}^{'}\in \mathbb{N}}exp(\frac{\epsilon}{2})exp(\frac{\epsilon q(t_1,\mathcal{N}^{'})}{2\Delta q})}{\sum_{\mathcal{N}^{'}\in \mathbb{N}}exp(\frac{\epsilon q(t_1,\mathcal{N}^{'})}{2\Delta q})}\right)                   \\
  &                                                 \leq exp(\frac{\epsilon}{2})\cdot exp(\frac{\epsilon }{2})\cdot \left(\frac{\sum_{\mathcal{N}^{'}\in \mathbb{N}}exp(\frac{\epsilon q(t_1,\mathcal{N}^{'})}{2\Delta q})}{\sum_{\mathcal{N}^{'}\in \mathbb{N}}exp(\frac{\epsilon q(t_1,\mathcal{N}^{'})}{2\Delta q})}\right)       \\
  &                                                 = exp(\epsilon)
\end{align*}
where $\mathbb{N}$ is the set of neighbor sets sampled from the target categories, $Pr(t_1, \mathbb{N})$ and $Pr(t_2, \mathbb{N})$ are the probabilities of sampling $\mathbb{N}$ from categories $t_1$ and $t_2$, respectively. Since both target categories have exactly the same users, and the sampling is independent on rating scores, it is noteworthy that we have $Pr(t_1, \mathbb{N}) = Pr(t_2, \mathbb{N})$.

So far, we can conclude that the first three steps of KDPCF satisfy differential privacy. Meanwhile, since the last step Top-m Recommendation is merely based on the output of the first three steps, it is actually a post-process. Therefore, due to the post-processing property of differential privacy, KDPCF satisfies differential privacy. This completes the proof. $\Box$
\end{proof}

\section{Performance Analysis and Evaluation}\label{sec:experiment}

In our experiments, we use the recommendation dataset, MovieLens, which contains 100,000 ratings by 943 users on 1682 items. The format of each record is (user ID, item ID, rating, timestamp), where the timestamp attribute is not used in the experiments. Each user rated at least 20 items and the rating range is 1-5. In order to verify the accuracy of the results, we divide the data into two parts: the training set and the testing set, with a ratio of 4:1. Finally, considering the randomness of differentially private algorithms, all experimental results take the average of 100 runs.

In the experiments, recall and precision are used to measure recommendation performances. In the context of recommendations, recall and precision are defined as follows.
\begin{equation}\label{equ:recall}
  Recall = \frac{\sum_{u \in U}|R_u\cap T_u|}{\sum_{u \in U}|T_u|}
\end{equation}
\begin{equation}\label{equ:precision}
  Precision = \frac{\sum_{u\in U}|R_u \cap T_u|}{\sum_{u \in U}|R_u|}
\end{equation}
where $U$ is the user set, $R_u$ is the recommendation list of user $u$ provided by recommendation schemes based on the training set, and $T_u$ is the rating list of user $u$ in the testing set.


In order to evaluate the recommendation performance, we set $k=2|U|/(C_{min}+C_{max})$ in our recommendation scheme, KDPCF, and compare it with the following two schemes in term of recall and precision.
\begin{itemize}
\item \textbf{CF}: The original user-based collaborative filtering recommendation scheme without differential privacy guarantees.
\item \textbf{DPCF}: The differentially private user-based collaborative filtering recommendation scheme, which is essentially the same as the DP-UR algorithm in literature \cite{zhu2016differential}, except that, for comparison fairness, we use the simple composition theorem instead of the advanced one when composition is needed.
\end{itemize}

We compare the recommendation performances among three schemes, CF, DPCF and KDPCF under the following conditions: (1) when recommendation list length $m$ varies; (2) when the neighbor set size $N$ varies; and (3) when the privacy budget $\epsilon$ varies. In default, we set recommendation list length $m=30$, the neighbor set size $N=30$, and the privacy budget $\epsilon=1$.

In Fig.~\ref{fig:comparison} (a) and (b), we can easily see that recall and precision show opposite trends as $m$ increases. This is because in Eq.~\eqref{equ:recall} and  \eqref{equ:precision}, increasing $m$ means raising $R_u$, and thus raising the recall but reducing the precision. Fig.~\ref{fig:comparison} (c) and (d) show that both recall and precision raise gently as the the neighbor set size $N$ increases. Similarly, Fig.~\ref{fig:comparison} (e) and (f) also show that both recall and precision increase gently as the privacy budget $\epsilon$ increases.

\begin{figure}[!htp]
    \begin{minipage}[t]{0.48\linewidth}
    \centering
    \includegraphics[height=5cm,width=6.5cm]{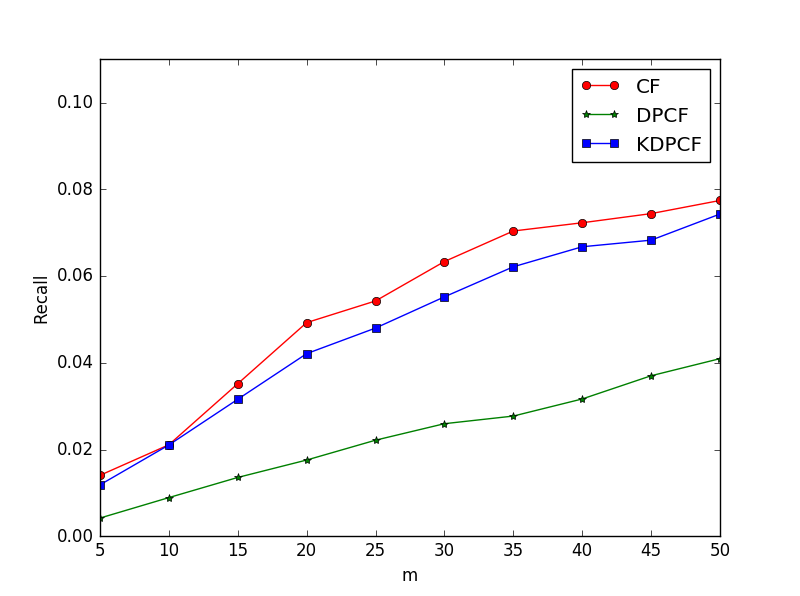}
    \caption*{(a)}
    \label{fig:side:a}
    \end{minipage}
    \begin{minipage}[t]{0.48\linewidth}
    \centering
    \includegraphics[height=5cm,width=6.5cm]{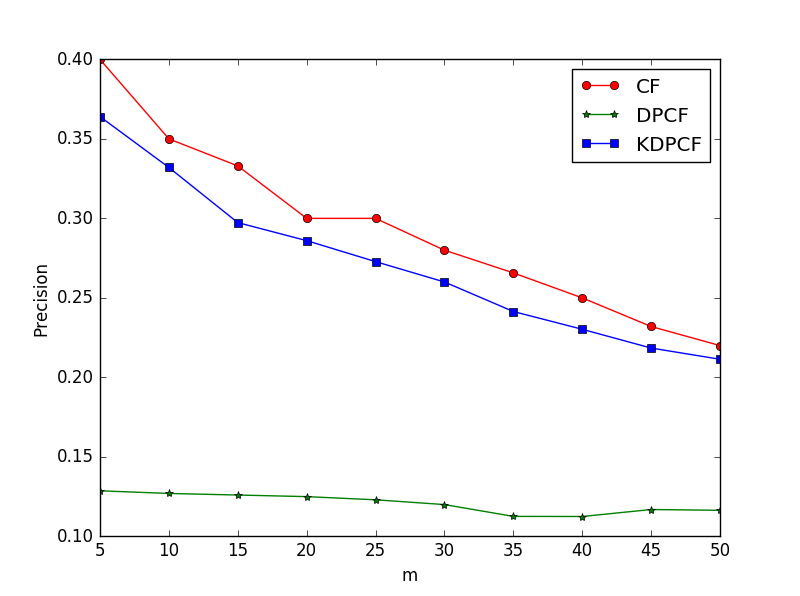}
    \caption*{(b)}
    \label{fig:side:c}
    \end{minipage}

    \begin{minipage}[t]{0.48\linewidth}
    \centering
    \includegraphics[height=5cm,width=6.5cm]{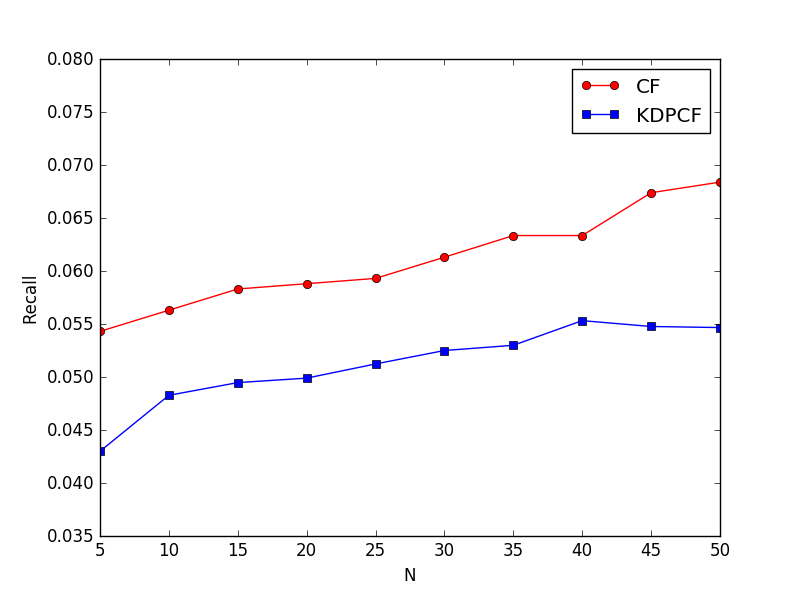}
    \caption*{(c)}
    \label{fig:side:c}
    \end{minipage}
    \begin{minipage}[t]{0.48\linewidth}
    \centering
    \includegraphics[height=5cm,width=6.5cm]{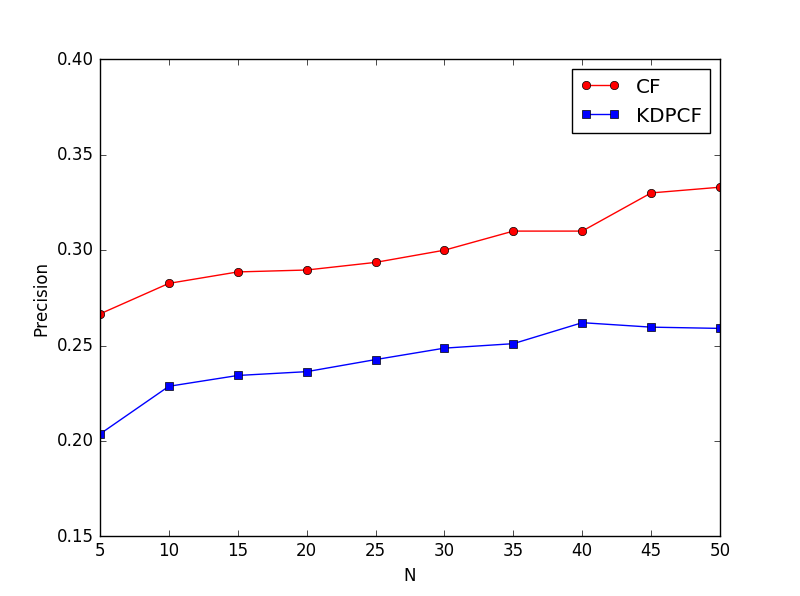}
    \caption*{(d)}
    \label{fig:side:d}
    \end{minipage}

    \begin{minipage}[t]{0.48\linewidth}
    \centering
    \includegraphics[height=5cm,width=6.5cm]{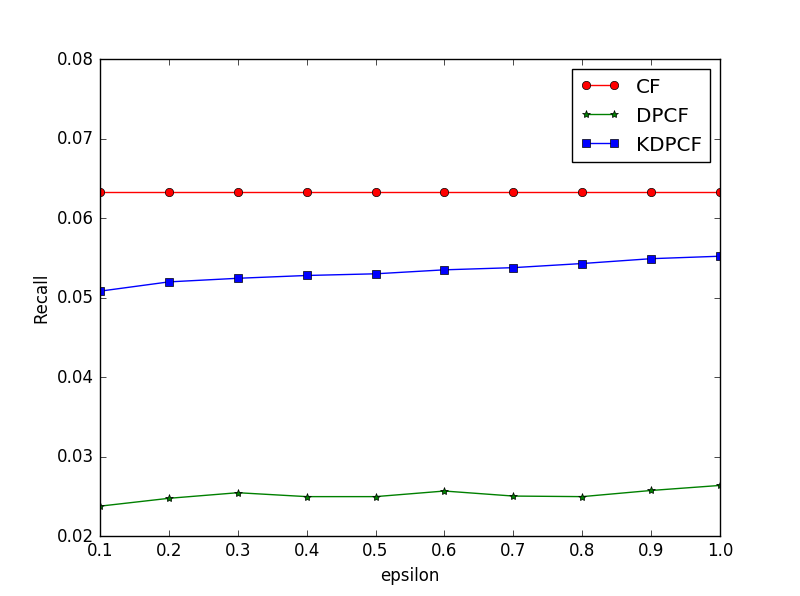}
    \caption*{(e)}
    \label{fig:side:b}
    \end{minipage}
    \begin{minipage}[t]{0.48\linewidth}
    \centering
    \includegraphics[height=5cm,width=6.5cm]{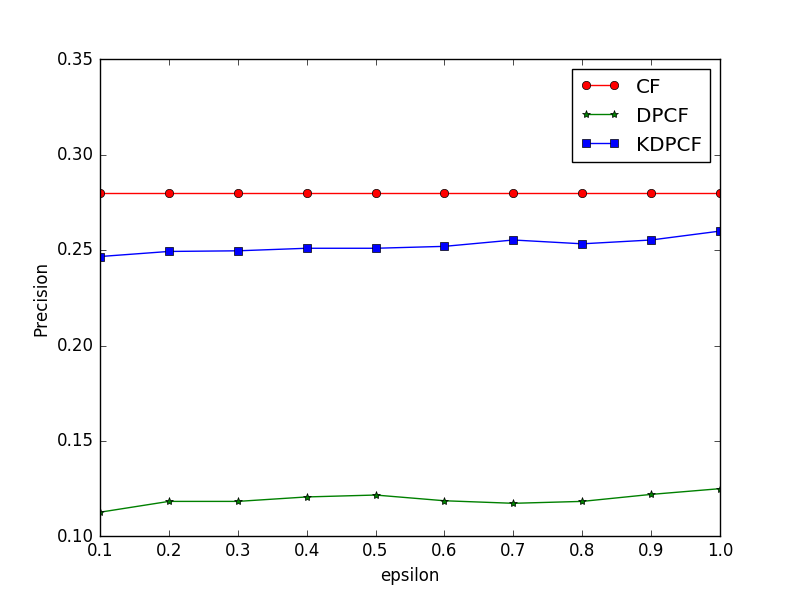}
    \caption*{(f)}
    \label{fig:side:d}
    \end{minipage}

    \caption{Performance Comparison in term of Recall and Precision}\label{fig:comparison}
\end{figure}

From all these figures, it is obvious to see that our scheme, KDPCF, is significantly superior to DPCF, the differentially private user-based collaborative filtering scheme without clustering, while slightly inferior to CF, the recommendation scheme without any differential privacy guarantees, when the recommendation list length $m$, privacy budget $\epsilon$, and the neighbor set size are the same. This demonstrates that reasonable clustering can greatly improved the recommendation performance of differentially private recommendation schemes.

\section{Conclusion}\label{sec:conclusion}

In this paper, we have proposed a differentially private user-based collaborative filtering recommendation scheme based on k-means clustering, KDPCF, addressing the recommendation performance degradation problem that arises in previous researches. Specifically, we apply k-means clustering to recommendations, and adjust the target category to an appropriate size, so that recommendations are limited within the well-sized target category, and thus a better balance between recommendation performance and user privacy can be obtained. We also employ the subset sampling to randomly select a neighbor set from the target category with a single application of the exponential mechanism, and perform recommendations based on the neighbor set, further improving the recommendation performance under the same privacy levels. Theoretically analysis shows that our scheme achieves differential privacy, and experimental evaluations demonstrate the high recommendation performance of our scheme.



\section*{References}

\bibliography{mybibfile}
\bibliographystyle{elsarticle-num}

\end{document}